\documentclass{article}
\usepackage{ijcai17}
\usepackage{times}
\usepackage{mathrsfs}
\RequirePackage[sectionbib,square]{natbib}%\bibliographystyle{abbrvnat}
\usepackage{amssymb,amsmath,verbatim}
\usepackage[pdftex]{graphicx}

\usepackage{microtype}

\usepackage[english, vlined, ruled, linesnumbered]{algorithm2e}
\usepackage{amsfonts,amsthm,booktabs,enumitem}
\usepackage{tikz}
\usetikzlibrary{arrows,shapes, calc, fit, positioning}
\usepackage{tkz-graph}
\usepackage{caption}
\usepackage{subcaption}

\renewcommand*{\ge}{\geqslant}
\renewcommand*{\le}{\leqslant}

\newcommand{\bbR}{\mathbb{R}}
\newcommand{\poly}{\mathrm{poly}}

\newcommand{\calP}{\mathcal{P}}
\newcommand{\calT}{\mathcal{T}}
\newcommand{\calU}{\mathcal{U}}

\newcommand{\calA}{\mathcal{A}}

\newcommand{\MMS}{\mathrm{mms}}
\theoremstyle{plain}
	  \newtheorem{theorem}{Theorem}[section]
	  
	  \newtheorem{lemma}[theorem]{Lemma}
	  \newtheorem{prop}[theorem]{Proposition}

\theoremstyle{definition}
	  \newtheorem{define}[theorem]{Definition}
	  
	  \newtheorem{example}[theorem]{Example}
	  
\theoremstyle{remark}
	  
\numberwithin{equation}{section}

\title{Fair Division of a Graph}

\author{
Sylvain Bouveret \\ LIG - Grenoble INP, France \\ sylvain.bouveret@imag.fr \And
Katar\'ina Cechl\'arov\'a \\ P.J. \v{S}af\'arik University, Slovakia \\ katarina.cechlarova@upjs.sk \AND
Edith Elkind \\ University of Oxford, UK \\ elkind@cs.ox.ac.uk \And
Ayumi Igarashi \\ University of Oxford, UK \\ ayumi.igarashi@cs.ox.ac.uk \And
Dominik Peters \\ University of Oxford, UK \\ dominik.peters@cs.ox.ac.uk
}

\begin{document}

\maketitle

\begin{abstract}
We consider fair allocation of indivisible items under an additional constraint:
there is an undirected graph describing the relationship between the items, and each agent's 
share must form a connected subgraph of this graph. This framework captures, e.g., fair 
allocation of land plots, where the graph describes the accessibility relation among the plots.
We focus on agents that have additive utilities for the items,
and consider several common fair division solution concepts, such as proportionality, envy-freeness
and maximin share guarantee. While finding good allocations according to these solution concepts
is computationally hard in general, we design efficient algorithms for special cases where
the underlying graph has simple structure, and/or the number of agents---or, less restrictively, 
the number of agent types---is small. In particular, despite non-existence results in the general case, we prove that for acyclic graphs 
a maximin share allocation always exists and can be found efficiently.
\end{abstract}

\section{Introduction}
The department of computer science at University X is about to move to a new building.
Each research group has preferences over rooms, but it would also be desirable
for each group to have a contiguous set of offices, to facilitate communication.
This situation can be seen as a problem of fair division (where agents are research groups 
and items are offices) with an additional connectivity requirement. This constraint could be captured by an undirected graph whose vertices are rooms (items) and there
is an edge between two vertices if the respective rooms are adjacent; each agent should 
obtain a connected piece of this graph.

In this paper, we introduce and study a formal model for such scenarios.
Specifically, we consider the problem of fair allocation of indivisible items 
in settings where there is a graph capturing the dependency relation between items, 
and each agent's share has to be connected in this graph. Besides the example in the first paragraph, 
our model captures a variety of applications, such as time-sharing 
a processor where tasks can be switched only at pre-defined times,
allocating a set of indivisible land plots, or assigning administrative duties
to members of an academic department, where there are dependencies among tasks
(e.g., dealing with incoming foreign students has some overlap with preparing 
study programmes in foreign languages, but not with fire safety).

\smallskip

\noindent{\bf Our contribution\ }
We propose a framework for fair division under connectivity constraints, and 
investigate the complexity of finding good allocations in this framework 
according to three well-studied solution concepts: proportionality, envy-freeness
(in conjunction with completeness), and maximin share guarantee.
We focus on additive utility functions.

For proportionality and envy-freeness, we obtain hardness results even for very simple graphs:
finding proportional allocations turns out to be 
NP-hard even for paths, and finding complete envy-free allocations
is NP-hard both for paths and for stars. Nevertheless, we also obtain some positive
results for these solution concepts. In particular, both proportional
and complete envy-free allocations can be found efficiently
when the graph is a path and
agents can be classified into a small number of {\em types},
where agents are said to have the same type
when they have the same preferences over items.\footnote{The same parameter was used by \citet{Branzei2016} to obtain 
%similar results to ours 
results
for maximizing social welfare;
%this is the first attempt to make use of agent types in the context of fair division
similar ideas have been used in the context of coalition formation \citep{shrot,azizdekeijzer}.}
If we assume that not just the number of player types, but the actual number of players
is small, we obtain an efficient algorithm for finding proportional allocations on arbitrary trees.

Recently, several papers have studied the concept of the \emph{maximin share guarantee} (MMS) \citep{budish}, 
which captures a desirable property of allocations that is easy to achieve for divisible items 
via cut-and-choose protocols. For indivisible goods, such allocations %are not guaranteed 
need not exist \citep{Procaccia14,KurokawaPW16}.
We prove a strong positive result for our setting: %we show that
an MMS allocation always exists if the
underlying graph is a tree, and %this allocation 
can be computed efficiently. 
Our algorithm is an adaptation of the classic last-diminisher procedure for the divisible case.
In contrast, we provide an example where the
underlying graph is a cycle of length 8 and there is no MMS
allocation. We believe that these results are useful for developing an
intuitive understanding of the concept of MMS; in particular, our example for the
cycle is much simpler than known examples of instances with no MMS
allocation in the absence of graph constraints.

\smallskip

\noindent{\bf Related work\ }
Fair allocation of indivisible items has received a considerable amount of attention 
in the (computational) social choice literature;
we refer the reader to a survey by \cite{bouveret-chap}. 
However, ours is the first attempt
to impose a graph-based constraint on players' bundles. In contrast, in the context of fair allocation
of {\em divisible} items (also known as cake-cutting) contiguity is a well-studied requirement.
For instance, \cite{Stromquist80} showed that  an envy-free division in which each player
receives a single contiguous piece always exists, but it cannot be obtained by a finite algorithm, 
even for three players \citep{Stromquist}. These results extend to equitable division 
with contiguous pieces \citep{CechlarovaDobosPillarova,AumannDomb2010,CechlarovaPillarova2012b}.
\cite{Bei12} consider fair allocations with contiguous pieces that approximately
maximize social welfare; \cite{AumannDH13} investigate a variant of this question 
without fairness constraints.

%From the extensive literature, let us cite at least some results. Assuming additive utilities, \cite{DemkoHill98} 
%prove that even for two players it is NP-hard to decide whether a proportional allocation exists and they provide 
%a lower bound for the ensured value, as a function of the number of players. \cite{Markakis2011WorstCase} 
%strengthen this guarantee and design a polynomial time algorithm to find allocations achieving this value. 
%\cite{Lipton04} study the optimization problem of finding an allocation with minimum possible envy. They show that 
%in the general case of any nonnegative monotone utilities the problem is not solvable or approximable in 
%polynomial time unless P = NP, but for additive utilities they provide several approximation algorithms. These 
%results were refined by \cite{Bouveret08Jair}, namely to show that to decide the existence of an efficient 
%envy-free allocation for additive utilities is NP-complete, even when all players have identical preferences, or 
%when utilities of players are taken from the set $\{0,1\}$. \cite{KBKZ-ADT09} prove that the problem of deciding 
%whether a given allocation is Pareto-optimal is coNP-complete, and that the problem of deciding whether there is a 
%Pareto-efficient and envy-free allocation is $\Sigma^p_2$ -complete.

\cite{ConitzerDS04} analyze a combinatorial auction setting that is somewhat similar to ours:
in their model, too, there is an undirected graph describing connections between items, 
and each agent's bid is connected with respect to this graph. They provide an algorithm for finding
an allocation that maximizes the social welfare and is in FPT with respect to the treewidth
of the item graph. \cite{AumannDH15} consider auctioning of a time interval, and obtain results both
for the case of pre-determined time slots (which corresponds to the model of \cite{ConitzerDS04}, 
with the item graph being a line) and for the case where the interval can be cut into arbitrary slots
(which is similar in spirit to cake-cutting). However, neither paper considers any fairness constraints.

Two very recent papers, like ours, combine graphs and fair division. 
\cite{chev17} consider the setting where agents are located in vertices of a graph.
Each agent has an initial endowment of goods and can trade with her neighbors in the graph. 
The authors ask what outcomes can be achieved by a sequence of mutually beneficial deals.
In the work of \cite{abebe}, the graph describes a visibility relation: agents
are located in vertices and an agent can only envy agents who are adjacent to her.
In contrast, in our model graphs represent the relationship between items rather than agents. 

%%%%%%%%%%%%%%%%%%%%%%%%%%%%%%%%%%%%%%%%%%%%%%%%%%%%%%%%%%%%%%%%%

\section{Our Model}
We study fair allocation of indivisible goods where each allocated bundle is connected
in an underlying graph.
\begin{define}
An instance of the {\em connected fair division problem (CFD)} is a triple $I=(G,N,\calU)$ where
\begin{itemize}
\item $G=(V,E)$ is an undirected graph,
\item $N=\{1,\dots,n\}$ is a set of {\em players}, or {\em agents},
\item $\calU$ is an $n$-tuple of utility functions 
$u_i: V \rightarrow \bbR_{\ge 0}$, where $\sum_{v \in V}u_i(v)=1$ for each $i \in N$.
\end{itemize}
We refer to elements of $V$ as {\em items}, and denote the number of items by $m$.
\end{define}
Note that when $G$ is a clique, CFD is equivalent to the classic problem of fair allocation
with indivisible items.

For each $X \subseteq V$, we set $u_i(X)=\sum_{v \in X}u_i(v)$, so valuations in this paper are always additive. Two players $i,j \in N$ are of the 
{\em same type} if $u_i(v)=u_j(v)$ for all $v \in V$. We denote the number of player types
in a given instance by $p$. 

An \emph{allocation} is a function $\pi:N \rightarrow 2^V$ assigning each player a bundle of items. An allocation 
$\pi$ is {\em valid} if for each player $i \in N$ the bundle $\pi(i)$ is connected in $G$ and no item is allocated twice, so that
$\pi(i) \cap \pi(j) =\emptyset$ for each pair of distinct players $i,j \in N$.
We say that a valid allocation $\pi$ is 
\begin{itemize}
%proportionality
\item
{\em proportional} if $u_i(\pi(i)) \ge \frac{1}{n}$ for all $i \in N$,
%envy-freeness
\item
{\em envy-free} if $u_i(\pi(i)) \ge u_i(\pi(j))$ for all $i,j \in N$, and 
%completeness
\item
{\em complete} if $\bigcup_{i\in N}\pi(i)=V$.
%pareto-optimality
%\item
%{\em Pareto-optimal} if we cannot improve the utility of one player without hurting another player, i.e., 
%if there is no valid allocation $\pi^{\prime}$ such that $u_i(\pi^{\prime}(i)) \ge u_i(\pi(i))$ 
%for all $i \in N$ and $u_j(\pi^{\prime}(j)) >  u_j(\pi(j))$ for some $j \in N$.
%\item
%{\em equitable} if $u_i(\pi(i)) = u_j(\pi(j))$ for all $i,j \in N$. 
\end{itemize}
Notice that an allocation that gives everybody an empty bundle is envy-free, 
so, to better express the idea of fairness, 
the requirement of envy-freeness is typically accompanied by 
completeness or Pareto-optimality.

We also consider {\em maximin share (MMS) allocations}~\citep{budish}, adapting
the usual definition to our setting as follows.  
Given an instance $I=(G,N,\calU)$ of CFD with $G=(V, E)$,
let $\Pi_n$ denote the space of all partitions of $V$ into $n$ connected pieces.
The {\em maximin share guarantee} of a player $i\in N$
is 
\[
\MMS_i(I) = \max_{(P_1, \dots, P_n)\in\Pi_n}\min_{j\in\{1, \dots, n\}} u_i(P_j). 
\]
Note that since we are only taking the maximum over connected partitions, these values may be lower than in the general setting without graph constraints.
A valid allocation $\pi$ is a {\em maximin share (MMS) allocation}
if we have $u_i(\pi(i))\ge \MMS_i(I)$ for each player $i\in N$.

%For computational purposes, we assume that all utility functions take values in $\mathbb Q$;
%also, to simplify presentation, when analyzing the complexity of our algorithms we usually 
%assume that all arithmetic operations have unit cost.
 
We consider the following computational problems that all 
take an instance $I=(G,N,\calU)$ of the connected fair division problem as input. For computational purposes, we assume that utility functions take values in $\mathbb Q$. Hardness results will use unary encodings of utility values (unless noted otherwise).
\begin{itemize}[leftmargin=5.5mm]
\item {\sc Prop-CFD}: Does $I$ admit a proportional valid allocation?
\item {\sc Complete-EF-CFD}: Does $I$ admit a complete envy-free valid allocation?
%\item {\sc PO-EF-CFD}: Does $I$ admit a Pareto-optimal envy-free valid allocation?
%\item {\sc Prop-Eq-CFD}: Does $I$ admit a proportional and equitable valid allocation?
\item{\sc MMS-CFD}: Does $I$ admit an MMS allocation?
\end{itemize}
We note that, given a valid allocation, one can check in polynomial time
whether it is proportional or envy-free; %or equitable; 
thus, the respective computational problems are in NP. 
%In contrast, the problem of deciding 
%whether a given allocation is Pareto-optimal is coNP-complete, 
%and the problem of deciding whether there is a 
%Pareto-efficient and envy-free allocation 
%is $\Sigma^p_2$-complete, even when the underlying graph is a clique \citep{KBKZ-ADT09}.

In what follows, we assume that the number of items $m$ is at least as large
as the number of players $n$, since otherwise at least one player gets nothing.
Also, given a positive integer $k$, 
we write $[k]$ to denote the set $\{1, \dots, k\}$.

\section{Proportionality}
We start with the bad news: it is hard to find a proportional allocation, 
even 
%if the graph $G$ is bipartite, 
%there are only two players and they are both of the same type (Theorem~\ref{thm:prop-2ag})
if the graph $G$ is a path. %(Theorem~\ref{thm:prop-path}).

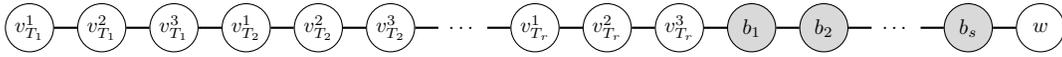
\begin{figure*}[htb]
\centering
\begin{tikzpicture}[scale=0.8, transform shape, every node/.style={minimum size=8mm, inner sep=1pt}]
	\node[draw, circle](0) at (0,0) {$v_{T_1}^1$};
	\node[draw, circle](1) at (1.2,0) {$v_{T_1}^2$};
	\node[draw, circle](2) at (2.4,0) {$v_{T_1}^3$};
	\node[draw, circle](3) at (3.6,0) {$v_{T_2}^1$};
	\node[draw, circle](4) at (4.8,0) {$v_{T_2}^2$};
	\node[draw, circle](5) at (6,0) {$v_{T_2}^3$};
	\node(6) at (7.2,0) {$\dots$};
	\node[draw, circle](7) at (8.4,0) {$v_{T_r}^1$};
	\node[draw, circle](8) at (9.6,0) {$v_{T_r}^2$};
	\node[draw, circle](9) at (10.8,0) {$v_{T_r}^3$};
	\node[draw, circle,fill=gray!30](10) at (12,0) {$b_1$};
	\node[draw, circle,fill=gray!30](11) at (13.2,0) {$b_2$};
	\node(12) at (14.4,0) {$\dots$};
	\node[draw, circle,fill=gray!30](13) at (15.6,0) {$b_{s}$};
	\node[draw, circle](14) at (16.8,0) {$w$};

	\draw[-, >=latex,thick] (0)--(1);	
	\draw[-, >=latex,thick] (2)--(1);	
	\draw[-, >=latex,thick] (2)--(3);	
	\draw[-, >=latex,thick] (4)--(3);	
	\draw[-, >=latex,thick] (4)--(5);	
	\draw[-, >=latex,thick] (6)--(5);	
	\draw[-, >=latex,thick] (6)--(7);	
	\draw[-, >=latex,thick] (8)--(7);	
	\draw[-, >=latex,thick] (8)--(9);	
	\draw[-, >=latex,thick] (10)--(9);	
	\draw[-, >=latex,thick] (10)--(11);		
	\draw[-, >=latex,thick] (12)--(11);		
	\draw[-, >=latex,thick] (12)--(13);		
	\draw[-, >=latex,thick] (14)--(13);							
\end{tikzpicture}
\caption{Graph constructed in the proof of Theorem~\ref{thm:prop-path}. 
\label{fig1}
}
\end{figure*}

%Path is hard%%%%%%%%%%%%%%%%%%%%%%%%%%%%%%%%%%%%%%%%
\begin{theorem}\label{thm:prop-path}
{\sc Prop-CFD} is {\em NP}-complete even if $G$ is a path.
\end{theorem}
\begin{proof}
We describe a polynomial-time reduction 
from the NP-complete problem {\sc Exact-3-Cover (X3C)} \citep{Garey79}. Recall that
an instance of {\sc X3C} is given by a set of elements $X=\{x_1,x_2,\dots, x_{3s}\}$ 
and a family $\calT=\{T_1,T_2,\dots, T_r\}$ of three-element subsets of $X$;
it is a `yes'-instance if and only if $X$ can be covered by $s$ sets from~$\calT$.
This problem remains NP-complete if for each element $x\in X$ its frequency
$p_x=|\{T\in \calT: x\in T\}|$ is at most $3$.

Consider an instance $J=(X, \calT)$ of X3C; for each $T\in\calT$, we denote the elements of $T$
by $x_T^1, x_T^2, x_T^3$. We construct an instance $I$ of  {\sc Prop-CFD} as follows.
There are three small vertices $v_T^1,v_T^2,v_T^3$ for each set $T\in \calT$, 
a set of $s$ big vertices $B=\{b_1,b_2,\dots, b_s\}$ 
and a dummy vertex $w$.
%$s+1$ dummy vertices $w_1,z_w,\dots, w_{s+1}$. 
The edges of $G$ are shown in Figure~\ref{fig1}.

There is one player $i_T$ for each $T\in \calT$, 
one player $i_x$ for each $x\in X$ and one dummy player $d$.
%$s+1$ dummy players $z_1,z_2,\dots,z_{s+1}$.
Hence the total number of players is $n=3s+r+1$. Define the utilities as: 
%$u_{i_T}(v) =1/(3n)$ if $v=v_T^k$, $u_{i_T}(v) =1/n$ if $v \in B$, $u_{i_T}(v) = (n-s-1)/n$ if $v \in B$
$$
u_{i_T}(v) = \left\{
\begin{array}{ll}
1/(3n) & \quad \mbox{\ if \ } v=v_T^k \phantom{\mbox{\ and\ } x\in T} \\ 
1/n &  \quad \mbox{\ if \ } v\in B\\
(n-s-1)/n & \quad \mbox{\ if \ } v=w \\ 
0 &\quad\mbox{\ otherwise}
\end{array}
\right.
$$
$$
u_{i_x}(v) = \left\{
\begin{array}{ll} 
1/n  & \quad \mbox{\ if \ } v=v_T^k \mbox{\ and\ } x\in T\\
\makebox[0pt][l]{($n-3p_x)/n$}\phantom{(n-s-1)/n} & \quad \mbox{\ if \ } v=w\\ 
0 & \quad \mbox{\ otherwise}
\end{array}
\right.
$$
$$
u_{d}(v)= \left\{
\begin{array}{ll}
\makebox[0pt][l]{$1$}\phantom{(n-s-1)/n\:} &\quad \mbox{\ if \ } v=w \phantom{\mbox{\ and\ } x\in T}\\
      0 &\quad \mbox{\ otherwise}
\end{array}
\right.
$$
By construction $u_i(V)=1$ for each $i\in N$. 
As player $d$ assigns a positive value to vertex $w$ only, 
she must receive this vertex in every proportional allocation. 
Given that $w$ is allocated to $d$, 
an allocation is proportional if and only if
each player $i_x$ receives a small vertex 
$v_T^k$ such that $x\in T$, and each player $i_T$ 
receives vertices $v_T^1, v_T^2, v_T^3$ (a triple interval) or a vertex from $B$.

Suppose that $J$ admits a cover $\calT'$ of size $s$. 
Let $\mu$ be a matching between $\calT'$ and $B$.
Assign intervals to players $i_x$ and $i_T$ as follows:
\begin{itemize}\itemsep0pt
\item 
For each $T\in \calT'$, player $i_T$ is assigned to vertex $\mu(T)$.
\item 
For each $T\notin \calT'$, player $i_T$ is assigned to the triple interval $v_T^1,v_T^2,v_T^3$.
\item 
Each player $i_x$ is assigned to the small vertex $v_T^k$ such that $x=x_T^k$ and $T\in \calT'$. 
\end{itemize}  
Then each player is assigned one connected piece and her value for that piece is at least $1/n$.

Conversely, suppose that $I$ admits a proportional valid allocation. As 
$|B|=s$, the number of $T$-players assigned to triple intervals is $r-s$. Hence, 
the number of triple intervals available for $x$-players is $s$, 
and the respective sets constitute an exact cover for $X$.
\end{proof}

%Note that the proof of Theorem~\ref{thm:prop-path} can be extended to show that {\sc Prop-CFD} 
%remains NP-hard on cycles; it suffices to add an edge between $T^1_1$ and $Z_{s+1}$.

\noindent 
In contrast, if $G$ is a star, finding a proportional allocation is easy. 
Our algorithm for this problem, as well as all other algorithms in this section, use matching techniques and can be adapted to also find valid allocations 
that maximize egalitarian welfare, i.e., the utility of the worst-off agent.
\begin{theorem}\label{thm:poly:stars:proportional}
{\sc Prop-CFD} is solvable in polynomial time if $G$ is a star.
\end{theorem}
\begin{proof}
Let $c$ denote the center of the star.
For each player $ i \in  N$ we check whether there is
a proportional valid allocation $\pi$ assigning $c$
to $i$. To this end, we create a bipartite graph $H=(Z,Z',L)$ with
$Z=N\setminus \{i\}$, $Z' = V\setminus\{c\}$ and $\{j,v\}\in L$
if and only if $u_j(v)\ge 1/n$; the weight of this edge is $u_{i}(v)$. 
Note that $|Z|\le |Z'|$. We say that a matching in $H$ is perfect 
if it matches all vertices in $Z$. Now, observe that
$I$ admits a proportional valid allocation
$\pi$ that assigns $c$ to $i$ 
if and only if $H$ admits a perfect matching $M$ of weight $w(M)\le (n-1)/n$;
indeed, assigning $c$ together with all remaining vertices 
to agent $i$ gives her a connected piece that she values as $1-w(M)\ge 1/n$.
It remains to observe that a minimum-weight perfect matching
can be computed in polynomial time [see e.g. Chapter $11$ of \citealp{Korte}]. 
\end{proof}

%%%%%%%%%%%%%%%%%%%%%%%%

\subsection{Bounded Number of Agent Types}
If the underlying graph is a path and all players are of the same type, 
then a simple greedy algorithm finds a proportional allocation 
(or reports that none exists) in linear time:
we build connected pieces one by one, by 
moving along the path from left to right and adding vertices to the current piece
until its value to a player reaches $1/n$; at this point we start building a new piece.
This procedure creates at most $n$ pieces;
a proportional valid allocation exists if and only if it
creates exactly $n$ pieces. 
More generally, if $G$ is a path and the number of agent types
is bounded by a constant, a simple dynamic
program can check the existence of a proportional allocation in polynomial time. 
A problem is {\em slice-wise polynomial} (XP) with respect to a parameter $k$ if each instance $I$ of this problem can be solved in time $|I|^{f(k)}$ where $f$ is a computable function.

\begin{theorem}\label{thm:XP:path}
{\sc Prop-CFD} is in {\em XP} with respect to the number of player types $p$
if $G$ is a path.
\end{theorem}
\begin{proof}
Let $G=(V, E)$, where $V = \{v_1,\dots,v_m\}$, $E=\{\{v_i, v_{i+1}\}: i\in [m-1]\}$. 
Suppose there are $n_t$ players of type $t$, for $t\in[p]$.
We say that a player is \emph{happy} if she gets a connected piece 
%whose value for her 
of value
at least $1/n$. Let $V_0 = \emptyset$ and $V_i=\{v_1, \dots, v_i\}$, $i > 1$.

For $i=0, \dots, m$, and a collection of indices
$j_1, \dots, j_p$ such that $0\le j_k\le n$ for each $k\in[p]$,
let $A_i[j_1, \dots, j_p]=1$ if
there exists a valid partial allocation $\pi$ of $V_i$ with
$j_k$ happy agents of type $k$, $k\in[p]$,
and let $A_i[j_1, \dots, j_p]=0$ otherwise.
Clearly, $A_0[j_1, \dots, j_p]=1$ if and only if $j_k=0$ for all $k\in[p]$.
For $i=1, \dots, m$, we have $A_i[j_1, \dots, j_p]=1$ if and only if 
there exists a value $s<i$ and $t\in[p]$
such that $A_s[j_1, \dots, j_t-1, \dots, j_p]=1$ and a player
of type $t$ values the set of items $\{v_{s+1}, \dots, v_i\}$
at $1/n$ or higher. 

A proportional allocation exists if $A_m[j_1, \dots, j_p]=1$
for some collection of indices $j_1, \dots, j_p$ such that
$j_t\ge n_t$ for all $t\in[p]$.
There are at most $(m+1)(n+1)^p$ values to compute, 
and each value can be found in time $O(mt)$ using unit cost
arithmetics. Thus, {\sc Prop-CFD} is in XP
with respect to $p$.
\end{proof}

%%%%%%%%%%%%%%%%%%%%%%%%%%%%%%%%%%%%%%%%%%%%%%
\subsection{Bounded Number of Agents}
If the number of agents $n$ is bounded by a constant and $G$ is a tree, 
then {\sc Prop-CFD} can be solved in polynomial time: we consider all possible ways
of partitioning the tree in $n$ non-empty connected pieces (a partition Scan be associated with 
a set of $n-1$ edges to be deleted, so there are ${m-1\choose n-1}$ partitions), and, for each partition,
we check if each player can be matched to a piece that she values at $1/n$ or higher
(by solving a simple bipartite matching problem). This shows that {\sc Prop-CFD} 
on trees is in XP with respect to the number of players. 
%Theorem~\ref{thm:XP:path} immediately implies that  {\sc Prop-CFD} can be solved in polynomial time if the number of agents $n$ is bounded by a constant and the graph is a path. 
A more careful argument shows that {\sc Prop-CFD} on trees, is, in fact, fixed parameter tractable with respect to this (weaker) parameter. A problem is {\em fixed parameter tractable} (FPT) with respect to a parameter $k$ if each instance $I$ of this problem can be solved in time $f(k)\poly(|I|)$ where $f$ is a function that depends only on $k$.

\begin{theorem}\label{thm:FPT:tree}
{\sc Prop-CFD}  is in {\em FPT} with respect to $n$
when $G$ is a tree.
\end{theorem}
\begin{proof}
We turn $G$ into a rooted tree by choosing an arbitrary node $r$ as the root.
Given a vertex $v$, we denote by $C(v)$ the set of children of $v$ 
and by $D(v)$ the set of descendants of $v$ (including $v$) in the rooted tree.

For each $v\in V$, $S \subsetneq N$, and $i \in N \setminus S$, 
let $\Pi_{i, v, S}$ be the set of all valid allocations $\pi:S\cup \{i\} \rightarrow 2^{D(v)}$ 
with $v \in \pi(i)$ and $u_j(\pi(j)) \ge {1}/{n}$ for all $j \in S$,
and define
\[
A_v[i,S]=\max\:\{u_i(\pi(i)): \pi\in \Pi_{i, v, S}\};
\]
by convention, $A_v[i,S]=-\infty$ when $\Pi_{i, v, S}$ is empty.
Note that we have a `yes'-instance of {\sc Prop-CFD}
if and only if $A_r[i, N\setminus\{i\}]\ge 1/n$ for some $i\in N$.

\iffalse
To complete the proof, we show that the values $A_v[i,S]$ can be efficiently
  computed in a bottom-up manner. We observe that for every
  internal vertex $v$ of the rooted tree, in each allocation
  $\pi\in \Pi_{i, v, S}$ the bundle of each player in $S$ is fully
  contained in a subtree $D(z)$ for some $z\in C(v)$. This induces a
  partition of players in $S$ into $|C(v)|$ groups. We can go through
  all possible partitions of $S$ with at most $|C(v)|$ parts and try
  to find a valid allocation compatible with this partition,
  maximizing the utility of each agent. Such an allocation can be
  found efficiently by solving a matching problem. 
\fi
 
We will now explain how to compute all values $A_v[i,S]$ in a bottom-up manner. When $v$ is a leaf of the rooted tree, we set $A_v[i,S]=u_i(v)$ if $S=\emptyset$ and $A_v[i,S]=-\infty$ otherwise.

Now, suppose that $v$ is an internal vertex of the rooted tree. 
If $S=\emptyset$, we have$A_v[i,S]=u_i(D(v))$, so from now on assume that $S \neq \emptyset$. We note that for each allocation $\pi\in \Pi_{i, v, S}$ the bundle of each player in $S$ is fully contained in a subtree $D(z)$ for some $z\in C(v)$. This induces a partition of players in $S$ into $|C(v)|$ groups. To find an allocation $\pi\in \Pi_{i, v, S}$ that maximizes the utility of player $i$, we go through all possible partitions of $S$ with at most $|C(v)|$ parts. For each such partition $\calP$, we try to find a valid allocation $\pi\in \Pi_{i, v, S}$ such that for each part $P\in\calP$ there is a unique child $z$ of $v$ such that the bundle of each player $i\in P$ is fully contained in $D(z)$; among such allocations, we pick one that maximizes $u_i(\pi(i))$.

To find such an allocation, we will construct an instance of the matching problem, which will decide which part of $\calP$ is assigned to which $z\in C(v)$. 
Thus, we construct a weighted bipartite graph $H_{\calP}=(Z,Z^\prime,L)$ as follows. 
We introduce one node $P$ for each part $P \in \calP$ and a set $X$ of $|C(v)|-|\calP|$ dummy nodes and set $Z=\calP\cup X$. We let $Z^\prime=C(v)$. By construction, $|Z|=|Z^\prime|$.

For every pair $(P, z)$ such that $P\in\calP$, $z \in Z^\prime$, and players in $P$ can be allocated items in $D(z)$, i.e., $A_z[i,P] \neq -\infty$ or $A_z[j,P \setminus \{j\}] \geq \frac{1}{n}$ for some $j \in P$, we construct an edge $\{P, z\} \in L$ with weight $w(P, z)$ 
corresponding to the maximum utility that player $i$ can receive from the set of items $D(z)$ under the constraint that players in $P$ 
obtain pieces of $D(z)$; specifically,
\[
w(P, z)=
\begin{cases}
A_z[i,P] & \mbox{if}~A_z[i,P] \neq -\infty,\\
0 & \mbox{otherwise}.
\end{cases}
\]
For every pair  $(x, z)$ with $x\in X$, $z \in Z^\prime$, we construct an edge $\{x, z\} \in L$ with weight $w(x, z)=u_i(D(z))$; 
matching $x$ to $z$ corresponds to assigning the items in $D(z)$ to player $i$. If $H_{\calP}$ admits a perfect matching, we set $w(\calP)$ to be the maximum weight of a perfect matching; otherwise, we set $w(\calP)=-\infty$.
Finally, we set 
\[
A_v[i,S]=\max \{\, w(\calP) : \calP~\mbox{is a partition of $S$} \land |\calP| \le |C(v)|\,\}.
\]
We omit the proof of the bound on the running time.
\end{proof}

We note that placing strong constraints on the underlying graph
is crucial for obtaining the easiness results in Theorems~\ref{thm:XP:path} and~\ref{thm:FPT:tree}. 
This is illustrated by the following simple proposition,
obtained by adapting a proof by \cite{DemkoHill98} for the standard
setting (with no graph constraints), which shows that the XP membership with respect to the number of players/types cannot be extended to arbitrary graphs.

%2 players is hard%%%%%%%%%%%%%%%%%%%%%%%%%%%%%
\begin{prop}\label{thm:prop-2ag}
When utilities are encoded in binary, 
{\sc Prop-CFD} is {\em NP}-complete even for $n=2$, $p=1$,
and even if the underlying graph $G$ is bipartite.
\end{prop}
\begin{proof}
We describe a polynomial-time reduction from {\sc Partition}. 
Recall that an instance of {\sc Partition} is given by 
a set of integers $J = \{a_i: i\in H\}$ such that $\sum_{i\in H}a_i=2k$. 
It is a `yes'-instance if and only if there exists 
a subset of indices $H'\subset H$ such that 
$\sum_{i\in H'}a_i=\sum_{i\in H\setminus H'}a_i=k$. 

Define an instance $I$ of {\sc Prop-CFD} as follows. 
Let $G=(V,E)$ where $V=\{v_i: i\in H\}\cup\{w_1,w_2\}$ 
and $E= \{\{v_i,w_1\}, \{v_i,w_2\}: i\in H\}$. 
There are two players with the same utility function
$u(v_i)=a_i/(2k)$ for $i\in H$ and $u(w_1)=u(w_2)=0$. 
Then $I$ admits a proportional valid allocation if and only if $J$ 
is a `yes'-instance of {\sc Partition}.  
\end{proof}

%%%%%%%%%%%%%%%%%%%%%%%%%%%%%%%%%%%%%%%%%%%%%%%%%%%%%%%%%%%%%%%%%%%%%%%%%%%%%%%%%%%%%%%%%%%%%%%

\section{Envy-freeness}
Envy-freeness turns out to be computationally more challenging than proportionality:
finding a complete envy-free allocation is NP-hard even if the underlying graph
is a star (for complete graphs, this result is shown by \citet{Lipton04}).
%star-envyfree is hard
\begin{theorem}\label{thm:NPc:stars:EF:Complete}
{\sc Complete-EF-CFD} is {\em NP}-complete even if $G$ is a star. 
\end{theorem}
\begin{proof}
Our hardness proof proceeds by a reduction from {\sc Independent Set}. 
Recall that an instance of {\sc Independent Set} is given by an undirected graph
$(W, L)$ and an integer $k$; it is a `yes'-instance if and only if $(W, L)$ contains 
an independent set of size $k$.
Given an instance $(W, L)$ of {\sc Independent Set}, we construct an instance
of {\sc Complete-EF-CFD} as follows. For each vertex $w\in W$
we create an item $w$ and a player $i_w$. 
Similarly, for each edge $\ell\in L$ we create an item $\ell$ and a player $i_\ell$.
We also create a set of dummy items $D$ with $|D|=k$, as well as an item $c$ and a player $i_c$.
The graph $G$ is a star with center~$c$ and set of leaves $W\cup L\cup D$.
Finally, define utility functions as follows.
\begin{itemize}
\setlength\itemsep{1pt}
\item
For each $w\in W$, we set $u_{i_w}(w) = 1/(k+1)$ and $u_{i_w}(d) = 1/(k+1)$ for each $d\in D$. 
\item
For each $\ell\in L$ with 
$\ell = \{x,y\}$, we set $u_{i_\ell}(\ell) = 3/7$,
$u_{i_\ell}(x) = u_{i_\ell}(y) = 2/7$. 
%(as $u_\ell(o_x) + u_\ell(o_y) > u_\ell(o_\ell)$, 
%$\ell$ would envy a player who receives both $o_x$ and $o_y$). 
\item
We set $u_{i_c}(c) = 1$.
\item
All other utilities are set to $0$.
\end{itemize}

We will now argue that there exists an independent set of size $k$ in the graph $(W,L)$ if and only if this 
instance of CFD admits a complete envy-free valid allocation.

Suppose there exists an independent set $X \subseteq W$ of size $k$. We construct an allocation $\pi$ 
as follows:
\begin{itemize}
\setlength\itemsep{1pt}
\item player $i_c$ receives $X\cup\{c\}$;
\item for $w \in W \setminus X$, player $i_w$ receives $w$;
\item for $w\in X$, player $i_w$ receives one item in $D$;
\item for $\ell\in L$, player $i_\ell$ receives $\ell$.
\end{itemize}
Clearly, $\pi$ is a complete valid allocation. It remains to show that $\pi$ is envy-free. 
%c
First, player $i_c$ does not envy any other player since she receives all her positive-utility items.
%v
Vertex players $\{i_w: w\in W\}$ receive utility $1/(k+1)$ in $\pi$; 
they could only envy someone who 
has multiple dummies, but no one does.
%e
Edge players $\{i_\ell: \ell\in L\}$ receive utility $3/7$ in $\pi$; 
the only way an edge player $i_\ell$ could envy another player is if that 
player got both items corresponding to endpoints of $\ell$. 
But the only player who receives more than one vertex item is player $i_c$ whose items correspond to an independent set. 
So no player envies anyone, and $\pi$ is envy-free.

Conversely, suppose that there is a complete envy-free valid allocation $\pi$. 
By completeness, $\pi$ allocates the central piece $c$ to some player. 
If $i_c$ does not receive $c$ then she would envy the player 
who receives it; so $c \in \pi(i_c)$. 
Thus, every other player receives at most one item.
Since $\pi$ is complete, this means that $i_c$ gets at least $k$ leaf items.
%TODO: DP: for camera-ready/journal version, give justification ^
Further, if $i_\ell$ does not receive $\ell$, 
she would envy the player who receives it,
so $\pi(i_\ell) = \{\ell\}$. Next, consider the bundle of player $i_c$.
If it contains more than one dummy item, vertex players would envy $i_c$.
Thus, it contains at least one item $w\in W$. If $\pi(i_c)$ also contains
a dummy item, $i_w$ would envy $i_c$, so $\pi(i_c)$ consists of $c$
and $k$ vertex items. Now, if there is an edge 
$\ell=(x, y)$ such that $x, y\in \pi(i_c)$, then
player $i_\ell$ envies $i_c$. Hence, $\pi(i_c)\setminus\{c\}$
forms an independent set of size $k$ in $(W, L)$.
\end{proof}

%This proof can be adapted to show that {\sc PO-EF-CFD}
%is NP-hard when $G$ is a star; moreover, by Proposition~\ref{prop:NP:stars:PO},
%when $G$ is a star we can check in polynomial time whether a given valid allocation
%is Pareto optimal. Combining these observations, we obtain the following theorem.

%\begin{theorem}\label{thm:NPc:stars:EF:PO}
%{\sc PO-EF-CFD} is {\em NP}-complete even if $G$ is a star. 
%\end{theorem}

We also obtain a hardness result for paths; the proof 
is similar to that of Theorem~\ref{thm:prop-path}.
%path-envyfree is hard

\begin{theorem}\label{t_NPC3}
The problem {\sc Complete-EF-CFD} is {\em NP}-complete even if $G$ is a path.
\end{theorem}
\begin{proof}
We shall show how to modify the polynomial transformation from {\sc X3C} 
provided in the proof of Theorem~\ref{thm:prop-path} to get the result.
The graph $G$ and the set of players $N$ are the same as in that proof.

For each $y_j$, let
$\mathcal{Z}_j$ be an arbitrary set of $s+1-p_j$ vertices among the
dummy vertices. Let us denote $U=3s(s+1)$ and let the utilities be:
\[
u_{y_j}(v) = \left\{
  \begin{array}{ll} 
    3s  &\ \qquad \text{\ if \ } v=T_i^k \text{\ and\ }  y_j=x_{ik}\\
    3s  &\ \qquad \text{\ if \ } v \in \mathcal{Z}_j\\ 
    0 &\ \qquad \text{\ otherwise}
  \end{array}
\right.
\]
\[
u_{t_i}(v) = \left\{
  \begin{array}{ll}
    s & \text{\ if \ } v=T_i^k \text{\ for\ } k=1,2,3\\ 
    3s &  \text{\ if \ } v=S_j \text{\ for\ } j=1,2,\dots, s\\
    0 &\text{\ otherwise}
  \end{array}
\right.
\]
\[
u_{z_k}(v)=\left\{
  \begin{array}{ll}
    3s &\quad \text{\ for \  } v = Z_1,Z_2,\dots,Z_{s+1}\\
    0 &\quad \text{\ otherwise}
  \end{array}
\right.
\]
It is easy to see that the total value of the whole cake is for each
player equal to $U$.　As  players $z_k$ assigns a nonzero value only
to vertices $Z_1,Z_2,\dots,Z_k$, they must be assigned to them in any complete envy-free and any proportional allocation
(otherwise if some other player $\alpha$ will get one of these vertices then one of $z_k$ players will  get a piece with utility 0 and will envy $\alpha$).

Observe further that in every proportional equitable allocation, each player $\alpha$ gets a utility of $3s$.  She must get at least this utility also in every envy-free allocation. Namely, if  $\alpha=y_j$ gets less than $3s$  then she will envy any dummy player receiving and if $\alpha = t_i$, then she will envy any
other player receiving an $S$-vertex.

The rest of the proof is now easy when we realize that each player gets a connected piece of value exactly $3s$.
%%%%%%%%%%%%%%%%%%%%%%%%%%%%%
Now, a player $y_j$ can get a connected piece of value at least $3s$
only if she is assigned to a small vertex $T_i^k$ corresponding to a
set that contains $x_j$ and player $t_i$ can get an interval of value
at least $3s$ only if she is assigned either to three vertices
$T_i^1, T_i^2,T_i^3$ (triple interval), or to one of the $S$-vertices.

Suppose that an exact cover $\calT'$ for $J$ exists.
We assign the intervals to players as follows:
\begin{itemize}\itemsep0pt
\item Each player $y_j$ is assigned to the small vertex $T_i^k$ if
  $x_j=x_{ik}$ and $T_i\in \calT'$.
\item Players $t_i$ corresponding to sets in $\calT'$ are assigned to
  vertices $S_k$, $k=1,2,\dots, s$ in an arbitrary order.
\item $t_i$ for $T_i\notin \calT'$ is assigned to the three small
  vertices $T_i^1,T_i^2,T_i^3$.
\end{itemize}  
It is easy to see that each player is assigned one connected piece and
the value of her piece is exactly $3s$. Since each player $y_j$
receives exactly one vertex, this share cannot be of value strictly
more than $3s$ for any other player (because no player values a single
vertex more than $3s$). Hence none can envy any player $y_j$. By the
same argument, none can envy any player $t_i$ receiving a vertex
$S_k$. The interval $T_i^1,T_i^2,T_i^3$ has a value $3s$ for $t_i$ and
$0$ for the other players $t_{i'}$. Hence the players $t_i$ cannot
envy each other. Now, interval $T_i^1,T_i^2,T_i^3$ has a value strictly
more than $3s$ for a player $y_j$ only if element $x_j$ appears more
than once in $T_i$, which could not happen by definition. Hence the
allocation is envy-free.

Conversely, suppose that a complete envy-free valid allocation in $I$
exists. As we have seen earlier, each player should receive at least
$3s$, which is possible only if each player $t_i$ either receives one
$S$-vertex or the triple interval $T_i^1,T_i^2,T_i^3$. As the number
of $S$-vertices is only $s$, the number of $t$-players assigned to
triple intervals is $r-s$. So the number of $T$-vertices available for
$y$-players is $3s$ and they constitute an exact cover for $J$.
%%%%%%%%%%%%%%%%%%%%%%%%%
\end{proof}
On the positive side, just as for {\sc Prop-CFD}, the problem {\sc Complete-EF-CFD} is also in XP with respect to 
the number of player types $p$, as long as $G$ is a path.

\begin{theorem}\label{thm:XP:path:EF}
{\sc Complete-EF-CFD} is in {\em XP} with respect to the number of player types $p$
if $G$ is a path.
\end{theorem}
\begin{proof}[Proof sketch]
Note that for an allocation to be envy-free, all pieces assigned to players of a given
type should have the same value to players of that type. When $G$ is a path, there are only $\smash{\binom{m+1}{2}} \le m^2$ different connected bundles. Hence there are at most $m^2$ many possibilities for the utility that a player of a given type can obtain in a valid allocation. 

Our algorithm works as follows. For each player type, 
it guesses the utility that players of that type assign to their pieces (this guessing
can be implemented by going over all possibilities, as there are at most $\smash{\left(m^2\right)^p}$ of them). It then proceeds
similarly to the dynamic programming algorithm in the proof of Theorem~\ref{thm:XP:path};
the only difference is that, when creating a piece of the form $\{v_{s+1}, \dots, v_i\}$ 
for a player of a given type, 
it checks that the utility of that player type for this piece is what it guessed for that type, 
and that other players' utility for this piece is at most their guessed utility.
\end{proof}

%\commentsylvain{I must admit that I don't understand this proof sketch (and I don't remember the complete proof). I can try to rephrase if someone sends the complete proof to me.}

%%%%%%%%%%%%%%%%%%%%%%%%%%%%%%%%%%%%%%%%%%%%%%%%%%%%%%%%%%%%%

\section{Maximin Share Guarantee}

\begin{figure*}[th]
	\begin{subfigure}[t]{0.33\textwidth}
		\centering
		\begin{tikzpicture}[scale=0.8]
		\tikzstyle{every node}=[fill=black!80,text=white,circle,inner sep=3pt] 
		\node(0) at (0,0) {$r$};
		\node[fill=black!65](1) at (-1,-1) {};
		\node[fill=black!65](2) at (1,-1) {};
		\node[fill=black!40](3) at (-1.5,-2) {};
		\node[fill=black!40](4) at (-0.5,-2) {};
		\node[fill=black!40](5) at (0.6,-2) {};
		\node[fill=black!40](6) at (1.45,-2) {};
		\node[fill=black!25](7) at (-1.7,-3) {};
		\node[fill=black!25](8) at (-1.3,-3) {};
		\node[fill=black!25](9) at (-0.7,-3) {};
		\node[fill=black!25](10) at (-0.3,-3) {};
		\node[fill=black!25](11) at (0.3,-3) {};
		\node[fill=black!25](12) at (1.3,-3) {};
		\node[fill=black!25](13) at (1.7,-3) {};
		
		\draw (0)--(1) (0)--(2); 
		\draw (1)--(3) (1)--(4); 
		\draw (2)--(5) (2)--(6); 
		\draw (3)--(7) (3)--(8); 
		\draw (4)--(9) (4)--(10); 
		\draw (5)--(11); 
		\draw (6)--(12) (6)--(13); 
		
		\draw [rounded corners=3mm]  (-0.15,-3.3)--(2.23,-3.3)--(1,-0.25)--cycle;
		
		\begin{scope}[xshift=-3]
		   %knife
		   \draw [rounded corners=0.2mm,fill=black!50] (0.6,-0.8)--(1.1,-0.1)--(0.3,-0.55)--cycle;
		   \draw [fill=black] (0.3,-0.55) -- (0.4,-0.65) -- (0.1,-0.85)-- (0,-0.75);
		\end{scope}
		\end{tikzpicture}
		\caption{The first player proposes a bundle.}
	\end{subfigure}%
	\begin{subfigure}[t]{0.33\textwidth}
		\centering
		\begin{tikzpicture}[scale=0.8]
		\tikzstyle{every node}=[fill=black!80,text=white,circle,inner sep=3pt] 
		\node(0) at (0,0) {$r$};
		\node(0) at (0,0) {$r$};
		\node[fill=black!65](1) at (-1,-1) {};
		\node[fill=black!65](2) at (1,-1) {};
		\node[fill=black!40](3) at (-1.5,-2) {};
		\node[fill=black!40](4) at (-0.5,-2) {};
		\node[fill=black!40](5) at (0.6,-2) {};
		\node[fill=black!40](6) at (1.45,-2) {};
		\node[fill=black!25](7) at (-1.7,-3) {};
		\node[fill=black!25](8) at (-1.3,-3) {};
		\node[fill=black!25](9) at (-0.7,-3) {};
		\node[fill=black!25](10) at (-0.3,-3) {};
		\node[fill=black!25](11) at (0.3,-3) {};
		\node[fill=black!25](12) at (1.3,-3) {};
		\node[fill=black!25](13) at (1.7,-3) {};
		
		\draw (0)--(1) (0)--(2); 
		\draw (1)--(3) (1)--(4); 
		\draw (2)--(5) (2)--(6); 
		\draw (3)--(7) (3)--(8); 
		\draw (4)--(9) (4)--(10); 
		\draw (5)--(11); 
		\draw (6)--(12) (6)--(13); 
		\draw [rounded corners=3mm,dashed]  (-0.15,-3.3)--(2.23,-3.3)--(1,-0.25)--cycle;
		
		%knife
		\draw [rounded corners=0.2mm,fill=black!50] (1.1,-1.7)--(1.6,-1)--(0.8,-1.45)--cycle;
		\draw [fill=black] (0.8,-1.45) -- (0.9,-1.55) -- (0.6,-1.75)-- (0.5,-1.65);
		\end{tikzpicture}
		\caption{Other players may diminish the bundle.}
	\end{subfigure}%
	\begin{subfigure}[t]{0.33\textwidth}
		\centering\qquad
		\begin{tikzpicture}[scale=0.8]
		\tikzstyle{every node}=[inner sep=3pt] 
		\node[fill=black!80,text=white,circle,inner sep=3pt] (0) at (0,0) {$r$};
		\node[fill=black!65,circle](1) at (-1,-1) {};
		\node[fill=black!65,circle](2) at (1,-1) {};
		\node[fill=black!40,circle](3) at (-1.5,-2) {};
		\node[fill=black!40,circle](4) at (-0.5,-2) {};
		\node[fill=black!40,circle](5) at (0.55,-2) {};
		\node[fill=black!40,circle](6) at (1.5,-2) {};
		\node[fill=black!25,circle](7) at (-1.7,-3) {};
		\node[fill=black!25,circle](8) at (-1.3,-3) {};
		\node[fill=black!25,circle](9) at (-0.7,-3) {};
		\node[fill=black!25,circle](10) at (-0.3,-3) {};
		\node[fill=black!25,circle](11) at (0.3,-3) {};
		\node[fill=black!25,circle](12) at (1.3,-3) {};
		\node[fill=black!25,circle](13) at (1.7,-3) {};
		
		\draw (0)--(1) (0)--(2); 
		\draw (1)--(3) (1)--(4); 
		\draw (2)--(5);
		%\draw [dashed] (2)--(6); 
		\draw (3)--(7) (3)--(8); 
		\draw (4)--(9) (4)--(10); 
		\draw (5)--(11); 
		\draw (6)--(12) (6)--(13); 
		\draw [rounded corners=3mm]  (0.8,-3.3)--(2.2,-3.3)--(1.5,-1.25)--cycle;
		
		\node at (2.5,-1.5) {$\cdots$};
		\end{tikzpicture}
		\caption{The last-diminisher receives the bundle.}
	\end{subfigure}%
	\caption{A discrete version of the last-diminisher method}
	\label{fig2}
\end{figure*}
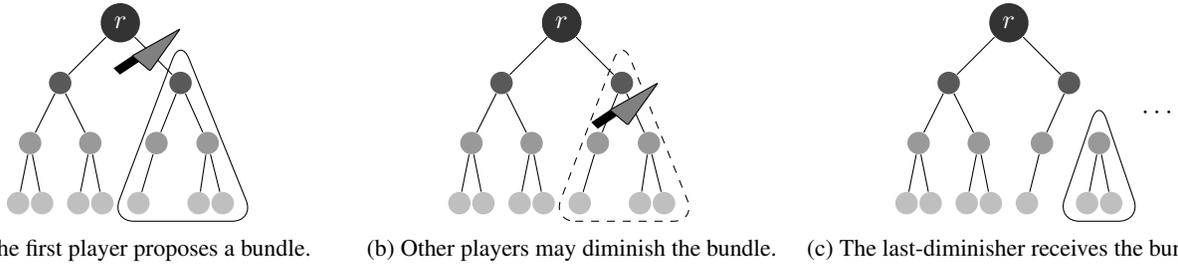

After \citet{budish} introduced the notion of an MMS allocation, it was open for some time whether every 
allocation problem (without connectivity constraints) admitted such an allocation. \citet{Procaccia14} 
found a counterexample. A family of more compact examples was found by \citet{KurokawaPW16}; these 
examples implicitly use an underlying grid graph; hence, for grid graphs, existence of MMS allocations is 
not guaranteed. Here, we show that for \emph{trees} an MMS allocation always exists. Our argument is 
constructive, and our algorithm corresponds to a discrete version of the last-diminisher method, which 
ensures proportionality while cutting a divisible resource [see, e.g., \citealp{Brams96}].
This method proceeds by letting one player identify a bundle of items. Every other player, in order, then 
has the option to \emph{diminish} this bundle by removing some of the items from it. The last player who 
chose to diminish is allocated the (diminished) bundle. The same procedure is then applied to divide the 
rest of the cake among the remaining $n-1$ players.
%This connection is interesting since the rationale behind the MMS criterion is to offer a discrete counterpart of proportionality in the cake-cutting setting.

We first describe an efficient procedure that guarantees each player 
a pre-specified level of utility.

\begin{prop}\label{thm:mms}
Let $I=(G, N, \mathcal{U})$ be an instance of CFD where $G$ is a tree and 
let $(q_i)_{i \in N}$ be an $n$-tuple of rational numbers. If $\MMS_i(I)\ge q_i$ for all $i \in N$, then there exists a valid allocation $\pi$ such that each player $i \in N$ receives the bundle of value at least $q_i$, i.e., $u_i(\pi(i))\ge q_i$. Moreover, one can compute such an allocation in polynomial time. 
\end{prop}
\begin{proof}
We will give an informal description of our recursive algorithm $\calA$ (Algorithm~\ref{alg:mms2}), 
followed by pseudocode. For each $X \subseteq V$, we let $G\setminus X$ denote the subgraph induced 
by $V\setminus X$; also, we denote the restriction of $u_i$ to $X$ by $u_i|_X$.

The algorithm first checks whether its input graph $G'$ 
has a value of at least $q_i$ for each player $i\in N'$; if this is not the case, it fails.
Then, if there is only one player, the algorithm simply returns the allocation that assigns all
items to that player. When there are at least two players, $\calA$ turns the graph into a rooted tree by choosing an arbitrary node as its root; denote by $D(v)$ the set of descendants of a vertex $v$ in this rooted tree. 
Then each player $i$ finds a vertex $v_i$ such that his value for
$D(v_i)$ is at least $q_i$, but for each child $w$ of $v$ his value
for $D(w)$ is less than $q_i$. The algorithm then allocates $D(v_i)$ to the {\em last-diminisher} $i$ 
whose vertex $v_i$ has minimal height (such a pair $(i, v_i)$ can be found by starting at the root of the tree and moving downwards). 
The player $i$ exits with the bundle $D(v_i)$, 
and the same algorithm $\calA$ is called on the remaining instance (see Fig.~\ref{fig2}).
\begin{algorithm}
\SetKwInOut{Input}{input} \SetKwInOut{Output}{output}
\SetKw{And}{and}
\SetKw{None}{None}
\caption{$\calA(I',(q_i)_{i \in N'})$}\label{alg:mms2}
\Input{$I'=(G', N', \calU')$ and $(q_i)_{i \in N'}$ where $G'$ is a subtree of $G$, $N'$ is a subset of $N$, and $u'_i=u_i|_{V'}$ for all $i\in N'$
}
\Output{A valid allocation $\pi$ such that $u_i(\pi(i)) \ge q_i$ for all $i \in N'$
}
	\If{$u'_i(V') < q_i$ for some $i \in N'$}{\Return{fail}}
	\ElseIf{$|N'|=1$}{
	\Return{$\pi$ where $\pi(i)=V'$ for $\{i\}=N'$}\;
	}
	\Else{
	Turn $G'$ into a rooted tree\;
	Find $i \in N'$ and $v_i \in V'$ such that $u'_{i}(D(v_i )) \geq q_i$, but $u'_{j}(D(w)) < q_j$ for each child $w$ of $v_i$ and each player $j \in N'$\;
	Set $I''=(G'\setminus D(v_i), N'\setminus\{i\},\calU'')$ where $\calU''$ is given by $u''_j=u'_j|_{V'\setminus D(v_i)}$ for all $j\in N'\setminus\{i\}$\;
	\If{$\calA(I'',(q_j)_{j \in N'\setminus \{i\}})$ does not fail}{
	Set $\pi^{\prime}\leftarrow \calA(I'',(q_j)_{j \in N' \setminus \{i\}})$\;
	Set $\pi(i)=D(v_i)$ and $\pi(j)=\pi^{\prime}(j)$ for each $j \in N'\setminus \{i\}$\;
	\Return{$\pi$}\;
	}
	}
\end{algorithm}

It is immediate that $\calA$ runs in polynomial time.
Let $I_n, \dots, I_1$ be the sequence of instances constructed by $\calA$ when called on $I$ and $(q_i)_{i \in N}$, 
where $I_k=(G_k, N_k, \calU_k)$ and $|N_k|=k$ (i.e., $I=I_n$). If $\calA$ does not fail
on any of these instances, then $\calA(I,(q_i)_{i \in N})$ returns a desired allocation: 
each agent is allocated a bundle that she values at least 
as highly as her given value $q_i$. 
We need to prove that none of the recursive calls fails.
To this end, we will prove the following lemma. 

\begin{lemma}\label{lem:mms-recursion}
$\MMS_j(I_k)\ge q_j$ for all $k\in [n]$ and all $j\in N_k$. 
\end{lemma}
\begin{proof}
The proof proceeds by backwards induction on $k$. For $k=n$
the statement of the lemma is true.
Suppose that the claim is true for some $k>1$; we will prove it for $k-1$.
Consider the player $i\in N_k\setminus N_{k-1}$, and let $D(v_i)$
be the bundle allocated to this player. 
For each player $j\in N_{k-1}=N_k\setminus\{i\}$ by the inductive hypothesis
we have $\MMS_j(I_k)\ge q_j$.
Consider a partition $\calP=(P_1, \dots, P_k)$ witnessing this;
$u_j(P_\ell)\ge q_j$  for each $\ell\in[k]$.
Assume without loss of generality that $v_i\in P_1$.
Then $D(v_i)$ is fully contained in $P_1$: if there is a vertex
$w$ in $D(v_i)\setminus P_1$, then the part of $\calP$
that contains $w$ is fully contained in a subtree rooted at a child of $v_i$, 
and hence the value of that part is strictly less than $q_j$, 
a contradiction. 

Now, if $P_1\setminus D(v_i)$ is not empty,
then it is a subtree of $G$, and there is another part $P\in\calP$
that is adjacent to $P_1\setminus D(v_i)$ in $G$. 
%, i.e., there is an edge $\{w, w'\}\in E$ with $w\in P$, $w'\in P_1\setminus D(v)$.
Therefore, $\calP'=(\calP\setminus \{P_1\})\cup \{P\cup(P_1\setminus D(v_i))\}$
is a partition of $G_{k-1}$ into $k-1$ connected components.
By construction, $u_j(P')\ge q_j$ for each $P'\in\calP'$, 
which proves that $\MMS_j(I_{k-1})\ge q_j$. 
\end{proof}
Now, consider what happens when $\calA$ is called on $I_k$ and $(q_i)_{i \in N_k}$ for some $k\in[n]$.
Let $G_k=(V_k, E_k)$. We have $u_i(V_k)\ge \MMS_i(I_k)\ge q_i$
for all $i\in N_k$, which implies that the algorithm does not fail.
This completes the proof.
\end{proof}

Proposition~\ref{thm:mms} relies on being given $(q_i)_{i\in N}$ as its input,
so we still need to show that MMS values on trees can be computed efficiently. 
%we can then invoke Proposition~\ref{thm:mms} with $q_i=\MMS_i(I)$ for each $i\in N$. 
It turns out that this can be accomplished by the same recursive algorithm.
We note that for the general problem (without graph constraints, or equivalently, 
on complete graphs), computing MMS values is NP-hard, 
though they can be well-approximated \citep{Woeginger97}.

\begin{lemma}\label{lem:mms1}
For an instance $I=(G, N, \mathcal{U})$ 
of CFD
where $G$ is a tree, and a player $i\in N$,
we can compute $\MMS_i(I)$ in polynomial time.
\end{lemma}
\begin{proof}
Fix a player $i\in N$. If $u_i(v)$ is represented as $x_v/y_v$, 
where $x_v$ and $y_v$ are integers (recall that $u_i$
is assumed to take values in $\mathbb Q$), 
set $u'_i(v)=u_i(v)\prod_{v\in V} y_v$.
Let $\MMS'_i(I)$ be the maximin share of player $i$ with respect to these new utilities.
Then $\MMS'_i(I)$ is an integer between $0$ and $mL^{m+1}$, where
$L=\max_{v\in V}\max\{x_v, y_v\}$ and
$$
\MMS_i(I)=\frac{1}{\prod_{v\in V} y_v}\MMS'_i(I)
$$ 
We now explain how to compute $\MMS'_i(I)$ in time polynomial in $n$, $m$, and $\log L$, i.e., in time

Calculating $\MMS_i(I)$ is the same as maximizing the worst payoff for the instance $I''$ 
where all players are copies of player $i$. 
%That is, $\MMS_i(I)$ is equal to the maximum value of $q$ such that $\calA$ does not fail on $(I'',(q,\ldots,q))$. 
%By scaling up the utilities, we can ensure that this value of $q$ is a not-too-large integer  (we omit the arithmetical details), and can therefore be found by binary search. 
That is, $\MMS'_i(I)$ is equal to the maximum positive integer of $q$ $\le mL^{m+1}$ such that $I''=(G,N'',\calU'')$ where $N''$ is a set of $n$ copies of $i$ and 
$\calU''$ is given by $u''_j=u'_i$ for all $j\in N''$ admits a valid allocation $\pi$ such that $u''_j(\pi(j)) \ge q$ for each copy $j \in N''$. 
Further, if such an allocation exists, then $\MMS_{j}(I'')\ge q$ for all $j \in N''$, and hence the call $\calA(I'',(q,\ldots,q))$ of the recursive algorithm in the proof of Theorem \ref{thm:mms} does not fail, returning a desired partition of $G$. 
Conversely, if $\calA(I'',(q,\ldots,q))$ does not fail, $I''$ clearly admits a valid allocation where each copy of $i$ gets a piece of value at least $q$. 
Thus, the maximum value of such $q$ can be found by binary search, which would require $O((m+1)\log L)$ calls to the subroutine $\calA$  and the running time of the subroutine itself is polynomial in $m$ and $\log L$. 
\end{proof}

\noindent It now follows from Proposition~\ref{thm:mms} and Lemma~\ref{lem:mms1} 
that for trees an MMS allocation can be computed efficiently.

%%%%%Computing mms allocations%%%%%%%%%%%%%%%
\begin{theorem}\label{the:mms2}
Every instance $I=(G, N, \calU)$ of CFD where $G$ 
is a tree admits an MMS allocation. 
Moreover, such an allocation can be computed in polynomial time.
\end{theorem}

\iffalse
Corollary~\ref{the:mms2} implies that an MMS allocation can always be found
efficiently when the graph is a path, which corresponds to a linear
discrete `cake'. For this case, the algorithm presented is a discrete
version of the Dubins--Spanier moving-knife procedure ensuring 
proportionality while cutting a continuous cake [see
e.g. \citealp{Brams96}]. Our algorithm moves a knife along the path, and
players shout `cut' whenever the left part of the path has a value of
at least $\MMS_i$. This connection is interesting since the rationale behind the
MMS criterion is to offer a discrete counterpart of proportionality in the cake-cutting setting.
\fi

The known examples of instances without MMS allocations are very intricate. Our graph-based setting 
allows for simpler constructions: our next example shows that 
an MMS allocation may not exist on a cycle of 8 vertices. We 
conjecture that this is the shortest cycle that admits such an example.
Our example is similar in spirit to an example for 2-additive utility functions
by \cite{BL-JAAMAS15}.

\begin{example}	\label{ex:mms}
Consider an instance $I=(G, N, \calU)$ of CFD 
where $G=(V, E)$ with $V=\{\, v_i \mid i=1,2,\ldots,8\,\}$, $E=\{\,\{v_i,v_{i+1}\} \mid 
i=1,2,\ldots,7\,\}\cup \{\{v_1,v_8\}\}$, $N=\{1,2,3,4\}$, and the utilities are given as follows.

\begin{center}
\begin{tabular}{ccccccccc} 
\toprule
 & $v_1$ & $v_2$ & $v_3$ & $v_4$ &  $v_5$ & $v_6$ &  $v_7$ &  $v_8$\\
\midrule
Players 1 \& 2 & 1 &4 & 4 & 1 & 3 & 2 & 2 & 3 \\
Players 3 \& 4 & 4 & 4 & 1 & 3 & 2 & 2 & 3 & 1 \\ 
\bottomrule
\end{tabular}
\end{center}
To normalize to $1$, each utility value above is divided by $20$.
Now, we have $\MMS_1(I)=\MMS_2(I)\ge 1/4$, as witnessed by the partition 
$P_1=\{\{v_1,v_2\},\{v_3,v_4\},\{v_5,v_6\},\{v_7,v_8\}\}$, which offers value $1/4$ for these players. 
Similarly, we have $\MMS_3(I)=\MMS_4(I) \ge 1/4$, as witnessed by 
the partition $P_2=\{\{v_2,v_3\},\{v_4,v_5\},\{v_6,v_7\},\{v_8,v_1\}\}$. These two partitions are illustrated below (note the cyclic shift):% in Figure~\ref{fig:empty}; note the clockwise shift.

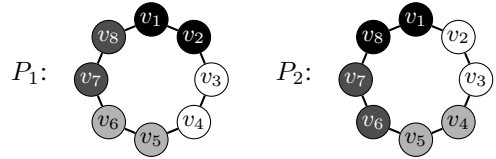
\begin{center}
  \parbox{0.4\columnwidth}{
    \centering
    $P_1$:\quad
    \begin{tikzpicture}[baseline]%[scale=0.4, transform shape]
      \footnotesize
      \def \radius {0.8cm}
      \node[inner sep=1pt, draw, circle](node1) at ({0}:\radius) {$v_{3}$};
      \node[inner sep=1pt, draw, circle,fill=black](node2) at ({45}:\radius) {\color{white}$v_{2}$};
      \node[inner sep=1pt, draw, circle,fill=black](node3) at ({90}:\radius) {\color{white}$v_{1}$};
      \node[inner sep=1pt, draw, circle,fill=black!70](node4) at ({135}:\radius) {\color{white}$v_{8}$};
      \node[inner sep=1pt, draw, circle,fill=black!70](node5) at ({180}:\radius) {\color{white}$v_{7}$};
      \node[inner sep=1pt, draw, circle,fill=black!30](node6) at ({225}:\radius) {$v_{6}$};
      \node[inner sep=1pt, draw, circle,fill=black!30](node7) at ({270}:\radius) {$v_{5}$};
      \node[inner sep=1pt, draw, circle](node8) at ({315}:\radius) {$v_{4}$};
      
      \draw[-, >=latex,thick] (node1)--(node2);
      \draw[-, >=latex,thick] (node3)--(node2);
      \draw[-, >=latex,thick] (node3)--(node4);
      \draw[-, >=latex,thick] (node5)--(node4);
      \draw[-, >=latex,thick] (node5)--(node6);
      \draw[-, >=latex,thick] (node7)--(node6);
      \draw[-, >=latex,thick] (node7)--(node8);
      \draw[-, >=latex,thick] (node1)--(node8);
    \end{tikzpicture}
  }
  \parbox{0.4\columnwidth}{
    \centering
    $P_2$:\quad
    \begin{tikzpicture}[baseline]%[scale=0.4, transform shape]
      \footnotesize
      \def \radius {0.8cm}
      \node[inner sep=1pt,draw, circle](node1) at ({0}:\radius) {$v_{3}$};
      \node[inner sep=1pt,draw, circle](node2) at ({45}:\radius) {$v_{2}$};
      \node[inner sep=1pt,draw, circle,fill=black](node3) at ({90}:\radius) {\color{white}$v_{1}$};
      \node[inner sep=1pt,draw, circle,fill=black](node4) at ({135}:\radius) {\color{white}$v_{8}$};
      \node[inner sep=1pt,draw, circle,fill=black!70](node5) at ({180}:\radius) {\color{white}$v_{7}$};
      \node[inner sep=1pt,draw, circle,fill=black!70](node6) at ({225}:\radius) {\color{white}$v_{6}$};
      \node[inner sep=1pt,draw, circle,fill=black!30](node7) at ({270}:\radius) {$v_{5}$};
      \node[inner sep=1pt,draw, circle,fill=black!30](node8) at ({315}:\radius) {$v_{4}$};
      
      \draw[-, >=latex,thick] (node1)--(node2);
      \draw[-, >=latex,thick] (node3)--(node2);
      \draw[-, >=latex,thick] (node3)--(node4);
      \draw[-, >=latex,thick] (node5)--(node4);
      \draw[-, >=latex,thick] (node5)--(node6);
      \draw[-, >=latex,thick] (node7)--(node6);
      \draw[-, >=latex,thick] (node7)--(node8);
      \draw[-, >=latex,thick] (node1)--(node8);
    \end{tikzpicture}
  }
\end{center}

Now, suppose towards a contradiction that the instance $I$ admits an MMS allocation $\pi$. 
Then $\pi$ has to allocate at least two vertices to each player, as no player values 
any single item at $1/4$ or higher.
This means that $\pi$ partitions the cycle into either $P_1$ or 
$P_2$. Suppose first that $\pi$ cuts the graph into $P_1$. Then, 
there is only one connected piece in $P_1$ that players 3 and 4 value at $1/4$ or higher,   
namely, $\{v_1,v_2\}$, so at least one of these players is allocated a 
piece whose value is less than his maximin share. 
A similar argument holds when $\pi$ cuts the graph into $P_2$. 
Therefore, there is no MMS allocation.\qed
\end{example}

\section{Conclusions and Future Work}
There are several exciting directions for the study of
  connected fair division of indivisible goods. For the solution
  concepts we have studied in this paper, one can ask whether certain
  graph classes yield better approximations than the general case,
  both in terms of existence guarantees and complexity results. 
  In particular, it would be interesting to obtain a characterization of graphs for which an MMS allocation is guaranteed to exist. 
  %we have seen that on trees, MMS allocations always exists, but the same is not true for large cycles and grid graphs. It would be interesting to move towards a characterization of graph classes which guarantee existence. The existing results may suggest that graphs with low circumference are particularly promising.
  There are also further solution concepts that we have not considered, most
  notably the maximum Nash welfare solution [see
  \citealp{Caragiannis16}], which could be studied in this context both
  from the axiomatic and the computational points of view. Another
  promising direction would be to extend the work to other
  preference representations, including ordinal preferences
  \citep{Aziz15}, or to chores instead of goods [e.g.,
  \citealp{Aziz17}]. Also, it would be interesting to obtain
  analogues of procedures such as sequential allocation and
  round-robin that respect the connectivity constraints and still
  produce desirable allocations. Finally, we may consider placing
  constraints on the `shapes' of players' pieces, e.g., by requiring
  that the size of each piece is large relative to its diameter;
  similar ideas have been recently explored by \cite{segal2015envy} in the context 
  of the land division problem (i.e., cutting a 2-dimensional cake).

% We believe that there are several further promising directions for the study of connected fair division of 
%indivisible goods. For the solution concepts we have studied in this paper, one can ask whether certain graph 
%classes can yield better approximations than the general case, both in terms of existence guarantees and 
%complexity results. There are also further solution concepts that we have not considered, most notably the maximum 
%Nash welfare solution (see \citealp{Caragiannis16}), where it would be interesting to see whether its axiomatic 
%properties can also be guaranteed in the graph-based setting, and whether restricted graph classes allow for 
%efficient computation.

% Acknowledge Dagstuhl and Budapest workshop.
\smallskip
\noindent
{\small
\textbf{Acknowledgements}\:
%We thank the anonymous reviewers for helpful comments that improved the presentation of the paper. 
This work was partly supported by the European Research Council (ERC) under grant number 639945 (ACCORD). 
Ayumi Igarashi is supported by an Oxford Kobe scholarship. 
Katar\'ina Cechl\'arov\'a is supported by grant APVV-15-0091 
from the Slovak Research and Development Agency. Sylvain Bouveret is partly supported by the
project ANR-14-CE24-0007-01 CoCoRICo-CoDec.
The authors are grateful to the organizers of the Dagstuhl Seminar 16232 ``Fair Division''
and Budapest Workshop on Future Directions in Computational Social Choice 
(supported by COST Action IC 1205), which enabled this collaboration. 
}

%% The file named.bst is a bibliography style file for BibTeX 0.99c
\bibliographystyle{abbrvnat}
%\bibliography{cakebib}

\end{document}